\documentclass{article}
\usepackage{amssymb}
\usepackage{xcolor}

\newtheorem{theorem}{Theorem}

\pdfoutput=1

\newtheorem{corollary}[theorem]{Corollary}

\newtheorem{definition}[theorem]{Definition}

\newtheorem{lemma}[theorem]{Lemma}

\newenvironment{proof}[1][Proof]{\noindent\textbf{#1.} }{\ \rule{0.5em}{0.5em}}
\usepackage{graphicx}
\graphicspath{%
    {converted_graphics/}
    {/}
}
\begin{document}

\begin{center}
{\LARGE The outcome of the restabilization process in matching markets  }%


\bigskip 

\bigskip

Mill\'{a}n Beatriz\footnote{%
San Luis Institute of Applied Mathematics (IMASL). National University of San Luis (UNSL) and the National Scientific and Technical Research Council (CONICET). Avenida Italia 1554. C.P.5700, San Luis, Argentina.
\par
Department of Mathematics, National University of San Juan (UNSJ). Av.
Ignacio de la Roza 230 (O). C.P.5400, San Juan,  Argentina. 

\par
e-mail: millanbetty2@gmail.com} {\LARGE \bigskip \bigskip }
\end{center}

{\small  Abstract: For a many-to-one matching model, we study the matchings obtained through the restabilization of stable matchings that had been disrupted by a change in the population. We include a simple representation of the stable matching obtained in terms of the initial stable matching (i.e., before being disrupted by changes in the population) and the firm-optimal stable matching\footnote{We used Lattice Theory to characterize the outcome of the restabilization process.}. We also describe the connection between the original stable matching and the one obtained after the restabilization process in the new market.


 }

\medskip

{\small Keywords: Stable matchings, Restabilization, Lattice.}

\medskip

{\small Subject Classification: C78, J73 }

\section{Introduction}

This  paper  studies  the  matchings  obtained through  the restabilization process of disrupted     many-to-one stable  matchings after a change in the population.   We provide various   characterizations of the obtained  stable  matchings based on  some  lattice-theoretic  results. We show that the outcome the restabilization process  is the meet of the set of the stable matchings unanimously preferred by workers to the initial stable matching.  Whenever the disruption is due to the retirement of some workers or the entry of new firms, the set of stable matchings unanimously preferred by workers to the initial stable matching is a nonempty sublattice of the set of stable matchings.  Thus it contains the matching meet, which is the worker-worst stable matching. Note that this matching is the outcome of the restabilization process.


\medskip

Blum et al. (1997) study in a one-to-one model the restabilization process triggered by the disruption of a pairwise stable matching due to the retirement of some workers or the creation of firms. Such a process leads to vacancy chains, since as one firm succeeds in filling its vacancy it may cause another firm to have one. In these cases they show that their modified version of the  Deferred Acceptance (DA) algorithm introduced by Gale and Shapley (1962) always reaches a pairwise stable matching. Furthermore, the stable matching which results is completely determined by the preferences of the agents, together with the particular firm quasi-stable matching at which the process starts. Firm quasi-stable matchings emerge as natural states in some markets, for example, in senior level labor markets, positions typically become available when current incumbents retire;  the resulting empty positions are often filled with candidates who are incumbents elsewhere.

\medskip

David Cantala (2003) studies, in many-to-one matching markets, the restabilization process of a stable matching disrupted by a change in the population, extending that  way the work of Blum et al. (1997), as he considers   that firms may hire many workers. He designs  the Set Offering (SO) Algorithm so as to mimic the restabilization process of a decentralized market, which always leads to a stable matching whenever the disruption is due to the opening of positions or the retirement of workers. When the disruption is due to the entrance of workers or the closure of positions, he constructs another algorithm which produces a stable matching. In this algorithm, unemployed workers make offers to firms.

\medskip


Wu and Roth (2018) show that the set of firm-quasi stable matchings forms
a lattice; and the set of stable matchings equals the set of fixed points of a Tarski operator on this lattice.\footnote{They use the term envy-free matchings to generalize the “simple” matchings studied by Sotomayor (1996) and the “firm quasi-stable” matchings studied by Blum et al. (1997).} Even though Blum et al. (1997) did not study the underlying lattice structure, they obtained some lattice-theoretic results on the one-to-one marriage model, some of which are extended to many-to-one settings in Cantala (2011) for firms with responsive preferences. We further investigate and generalize some of their conclusions in many-to-one model with substitutable and q-separable preferences. The extension of our analysis is two-fold. Firstly, we prove some technical results regarding the set of firm quasi-stable matchings that concern the lattice operators. These results are key for many results of this paper. Secondly, we characterize the outcome matching of the restabilization process, using results from the relation of the firm quasi-stable input matching for the SO Algorithm and the corresponding output. Observe that an implication of this characterization is that it allows us to explore the connection between the original stable matching and the one obtained after the restabilization in the new market.



\medskip
In the next section we describe  the formal
matching model, and reviews some results on stable matchings. Section 3 introduces firm quasi-stable matchings, and we  prove  some  technical  results. Section 4 studies restabilization after change in the population. Section 5 explores  the connection between the original stable matching and the one obtained after the restabilization process in the new market.

\section{Preliminaries}

The many-to-one bilateral matching market possesses two disjoint sets of agents (two-sided many-to-one matching model), the $n$ \textit{firms} $\mathcal{F}$ set and the $m$ \textit{workers} $\mathcal{W}$ set.  Each firm $f\in \mathcal{F}$ has a strict, transitive, and complete preference relation $P\left( f\right) $ over the set of all $\mathcal{W}$ subsets, and each worker has a strict, transitive, and complete preference relation $P\left(w\right) $ over $\mathcal{F}\cup \emptyset $.   {Preferences profiles} are $\left( n+m\right) $-tuples of preference relations and they are represented by $\mathcal{P}=\left(P(f_{1}),...,P(f_{n});P(w_{1}),...,P(w_{m})\right) $. Given a $\mathcal{P}$ preference profile, then the many-to-one bilateral matching market is the triplet $(\mathcal{F},\mathcal{W},\mathcal{P})$.

\medskip

The assignment problem consists of matching workers with firms maintaining the bilateral nature of their relationship and allowing for the possibility that both, firms and workers, may remain unmatched. Formally,

\begin{definition}
A \textit{matching} $\mu $ is a mapping from the set $\mathcal{F}\cup \mathcal{W}$ into the set of all subsets of $\mathcal{F}\cup \mathcal{W}$ such that for all $w\in \mathcal{W}$ and $F\in \mathcal{F}$:

\begin{enumerate}
\item Either $|\mu \left( w\right) |=1$ and $\mu \left( w\right) \subseteq \mathcal{F}$ or else $\mu \left( w\right) =\emptyset .$

\item $\mu \left( F\right) \in 2^{\mathcal{W}}$.

\item $\mu \left( w\right) =f$ if and only if $w\in \mu \left( f\right) .$\footnote{We will often abuse notation by omitting the brackets to denote a set with a unique element. For instance here, we write $\mu \left( w\right) =f$ instead of $\mu \left( w\right) =\left\{ f\right\} $.}
\end{enumerate}
\end{definition}

Criterion $1$ indicates that a worker is either matched to a firm or remains single. Criterion  $2$ shows that a firm is either matched to a subset of workers or remains single. Lastly, criterion $3$ states that the relationship is reciprocal.  $\mathcal{M}(\mathcal{F},\mathcal{W},\mathcal{P})$ denotes all of the possible matchings in $(\mathcal{F},\mathcal{W},\mathcal{P})$.

\medskip

We are following the convention of extending preferences from the original sets $(2^\mathcal{W}$ and $\mathcal{F} \cup \emptyset)$ to the set of matchings. However, we now have to consider weak preference relations since two matchings  may associate to an agent the same partner. These preference relations will be denoted by $R(f)$ and $R(w)$. For instance, to say that all firms prefer $\mu_1$ to another matching  $\mu_2$ means that for every $f\in  \mathcal{F}$ we have that $\mu_1(f) R(f) \mu_2(f)$ 
(that is, either $\mu_1(f) =\mu_2(f)$ or else $\mu_1(f) P(f) \mu_2(f)$). 

\medskip

We define the \textit{unanimous} partial orders  $\succeq_{\mathcal{F}}$ and $\succeq _{\mathcal{W}}$ in $\mathcal{M}(\mathcal{F},\mathcal{W},\mathcal{P})$ as follows:    
\begin{center}
$\mu _{1}\succeq _{\mathcal{F}}\mu _{2}\Leftrightarrow \mu_1(f) R(f) \mu_2(f)$  for all $f\in \mathcal{F}$ \\
$\mu _{1}\succeq _{\mathcal{W}}\mu _{2}\Leftrightarrow \mu_1(w) R(w) \mu_2(w)$  for all  $w\in \mathcal{W}$ 
\end{center}
We sometimes add superscripts and write, for example $\succeq_{\mathcal{F}}^{\mathcal{P}} $ o  $\succeq_{\mathcal{W}}^{\mathcal{P}} $  to emphasize dependence on particular preferences.

\medskip
Given a preference relation of a firm $P\left( f\right) $ the subsets of workers preferred to the empty set by $f$ are called acceptable. Therefore, we are allowing that firm $f$ may prefer not hiring any workers rather than hiring unacceptable subsets of workers. Similarly, given a preference relation of a $P\left( w\right) $ worker, the firms preferred by $w$ to the empty set are called acceptable. In this case we are allowing that worker $w$ may prefer to remain unemployed rather than working for an unacceptable firm. A pair  $(S,f)$ with $S\subseteq \mathcal{W}$ and $f\in \mathcal{F}$ is called acceptable coalition if $f$ is acceptable all workers $w$ in $S$  and $S$ is   acceptable for $f$. We denote by $A(\mathcal{F},\mathcal{W},\mathcal{P})$ the set of all coalitions acceptable  of market $(\mathcal{F},\mathcal{W},\mathcal{P})$.
Given a $S\subseteq \mathcal{W}$ set, let $Ch\left( S,P\left( f\right) \right) $ denote firm $f$'s most-preferred subset of $S$ according to its preference ordering $P\left(f\right) $ and we refer to this set as  \textit{choice}. 
A matching $\mu $ is \textit{blocked by a worker }$w$ if $%
\emptyset P\left( w\right) \mu \left( w\right) $; that is, worker $w$ prefers being unemployed rather than working for firm $\mu \left( w\right) $. Similarly, $\mu $ is \textit{blocked by a} \textit{firm} $f$ if $\mu
\left( f\right) \neq Ch\left( \mu \left( f\right) ,P\left( f\right) \right) $. We say that a matching is \textit{individually rational} if it is not blocked by any individual agent. A matching $\mu $ is \textit{blocked by} \textit{a worker-firm pair }$\left( w,f\right) $ if $w\notin \mu \left(f\right) $, $w\in Ch\left( \mu \left( f\right) \cup \left\{ w\right\},P\left( F\right) \right) $, and $fP\left( w\right) \mu \left( w\right) $; that is, if they are not matched through $\mu $, firm $f$ wants to hire $w$, and worker $w$ prefers firm $F$ rather than firm $\mu \left( w\right) $.

\begin{definition}
A matching $\mu $ is \textbf{stable} if it is not blocked by any individual agent or any firm-worker pair.
\end{definition}

We denote by $S(\mathcal{F},\mathcal{W},\mathcal{P})$ the set of stable matchings of market $(\mathcal{F},\mathcal{W},\mathcal{P})$. There are preference profiles in which the set of stable matchings is empty. These examples share the feature that at least one firm regards a subset of workers as  complements. This is the reason why the literature has focused on the restriction where workers are regarded as substitutes. 
The objective of substitutability condition is to make the hiring of a worker independent of the hiring of other workers. 
\footnote{Kelso and Crawford (1982) were the first to use this property (under the name of ``gross substitutability condition'') in a cardinal matching model with salaries.}

\begin{definition}
A firm $f$'s preference relation $P\left( f\right) $ satisfies \textbf{substitutability} if for any set $S$ containing workers $w$ and $\bar{w}$ $ \left( w\neq \bar{w}\right) $, if $w\in Ch\left( S,P\left( f\right) \right) $ then $w\in Ch\left( S\backslash \left\{ \bar{w}\right\} ,P\left( f\right)
\right) $.
\end{definition}

A preference profile $\mathcal{P}$ is \textit{substitutable} if for each firm $f$, the preference relation $P\left( f\right) $ satisfies substitutability. Kelso and Crawford (1982) shows that if all firms have substitutable preferences then the set of stable matchings is non-empty, and  firms unanimously agree that a stable matching $\mu_\mathcal{F}$ is the best stable matching. Roth (1984) extends these results and shows that if all firms have substitutable
preferences then workers unanimously agree that a stable matching $\mu_\mathcal{W}$ is the best stable matching, and  the optimal stable matching for one side is the worst stable matching for the other side. That is, $S(\mathcal{F},\mathcal{W},\mathcal{P})\neq \emptyset$; and for all $\mu\in(\mathcal{F},\mathcal{W},\mathcal{P})$ we have that $\mu_\mathcal{F}  \succeq_\mathcal{F} \mu  \succeq_\mathcal{F} \mu_\mathcal{W} $ and $\mu_\mathcal{W}  \succeq_\mathcal{W} \mu  \succeq_\mathcal{W} \mu_\mathcal{F}$.

\medskip

We will assume that firms' preferences satisfy a further restriction called $q-$separability.\footnote{See Mart\'{\i}nez et al. (2000) and (2001) for a detailed discussion of this restriction.} This is based on two ideas. First, separability, which says that the division between \textit{good} workers ($wP\left( f\right) \emptyset $) and \textit{bad} workers ($\emptyset P\left(f\right) w$) guides the ordering of subsets in the sense that adding a good worker leads to a \textit{better} set, while adding a bad worker leads to a \textit{worse} set. Second, each firm $f$ has in addition a maximum number of positions to be filled:\ its quota $q_{f}$. This limitation may arise from, for example, technological, legal, or budgetary reasons.
Formally,

\begin{definition}
A firm $f$'s preference relation $P(f)$  is $q_{f}\mathbf{-}$\textbf{separable} if: (a) for all $S\subsetneq \mathcal{W}\,$ such that $\left| S\right| <q_{f}$ and $w\notin S$ we have that $\left(S\cup \{w\}\right) P(f)S$ if and only if $wP(f)\emptyset $, and (b) $\emptyset P\left( f\right) S$ for all $S$ such that $\left| S\right| >q_{f}$.
\end{definition}

We will denote by $q=\left( q_{f}\right)_{f\in \mathcal{F}}$ the list of quotas and we will say that a preference profile $\mathcal{P}$ is $q-$\textit{separable }if each $P(f)$ is $q_{f}-$separable.

\medskip

As we study the properties of firm quasi-stable matchings, it will be useful to recall the following properties of stable matchings. From now on we will assume that firms have $q$-separable \textit{and}
substitutable preferences.  Mart\'{\i}nez et al (2001) establishes the fact that, under these assumptions, the set of stable matching has a lattice structure.
Given matchings $\mu _{1}$ and $\mu _{2}$, \textit{only }asking each worker to select the best  firm matched with them through $\mu _{1}$ and $\mu _{2}$.  In this way, we define the pointing function  $\mu _{1}\underline{\vee }_{\mathcal{W}}\mu _{2}$ on $\mathcal{F}\cup \mathcal{W}$ by: 
 
\medskip

$\mu_1 \veebar_\mathcal{W} \mu_2(w)=\left\{ \begin{tabular}{l} $\mu_1(w)$ \ \  if  \ \  $\mu_1(w) P(w) \mu_2(w)$ \\
 $\mu_2(w)$ \ \ \ \ \  otherwise \end{tabular} \right.$ for all $w\in \mathcal{W}$ and 

\medskip

$\mu_1\veebar_\mathcal{W} \mu_2(f)=\{w: \mu_1\veebar_W \mu_2(w)=f\}$  \ for all $f \in \mathcal{F}$.

\medskip

\noindent Symmetrically, define the pointing function $\mu _{1}\underline{\wedge }_{\mathcal{W}}\mu _{2}$ on $\mathcal{F}\cup \mathcal{W}$ by matching each worker with their worst firm and each firm with the corresponding set of workers that selected it, if any.

\begin{theorem}\label{reti}
Let $\mathcal{P}$ be a profile of substitutable and $q-$separable preferences. Then, $(S(\mathcal{F},\mathcal{W},\mathcal{P}),\succeq_\mathcal{W}, \wedge, \vee)$ is a lattice, where $\wedge=\barwedge_{\mathcal{W}}$ and $\vee= \veebar_{\mathcal{W}}$.
\end{theorem}

The following theorem, which has been proved by   Mart\'{\i}nez et al. (2000), states that  the number of workers assigned to a firm through stable matchings is the same; if the firm does not complete its quota under some stable matching, then it gets the same set of workers at any stable matching.

\begin{theorem}\label{single} 
Let $\mathcal{P}$ be a profile of substitutable and $q-$separable preferences. Then, all pairs $\mu, \mu'\in S(\mathcal{F},\mathcal{W},\mathcal{P})$, and all $f \in \mathcal{F}$:
\begin{enumerate}
\item $|\mu(f)|= |\mu'(f)|$. 
\item If $\mu(f)< q_f$, then $\mu(f)= \mu'(f)$.
\end{enumerate}
\end{theorem}

\section[Firm quasi-stable matching ]{Firm quasi-stable matching}

In this section we describe a class of matchings we call firm quasi-stable and we prove some technical results, which are   key for this paper. Firm quasi-stability was first introduced by Sotomayor (1996) as part of a new  proof of the existence of stable matchings. This  concept was also   analyzed by Blum et al. (1997) in one-to-one matching models and   Cantala (2003) extends to many-to-one models the definition of a firm quasi-stable matching. These matchings  can arise from the disruption of a stable matching due to changes in the population (retirement of workers or the creation of positions).

\medskip

Let $\mu$ be a matching, then we denote $W_{f,\mu}$ the set of workers who prefer firm $f$ to their match under $\mu$, that is, $W_{f,\mu}=\{ w\in \mathcal{F}: fP(w)\mu(w) \}$

\begin{definition}
A matching $\mu$ is firm quasi-stable if it is  individually rational and for every  $f \in \mathcal{F}$, $S\subseteq W_{f,\mu}$ and $S\not=\emptyset$, $\mu(f)\subseteq Ch(\mu(f)\cup S,P(f))$.
\end{definition}

This means that even if a firm is involved in a blocking coalition $(S,f)$, $f$ does not fire any worker when selecting their most preferred set of workers from those in $\mu(f)\cup S$. We denote by $FQS( \mathcal{F}, \mathcal{W},\mathcal{P})$ the set of firm quasi-stable matchings of market $( \mathcal{F}, \mathcal{W},\mathcal{P})$

\medskip

The next theorem  shows that the  binary operator $\underline{\vee }_{\mathcal{W}}$ preserves stability, when just one of the two underlying matchings is stable  and the other one is  firm-quasi-stable.

\begin{theorem}\label{s}
Let  $\mu_1\in S(\mathcal{F},\mathcal{W},\mathcal{P})$ y $\mu_2\in FQS( \mathcal{F},\mathcal{W},\mathcal{P})$. Then $\mu_1 \veebar_{W} \mu_2 \in S(\mathcal{F},\mathcal{W},\mathcal{P})$.
\end{theorem}

\begin{proof}
The individual rationality of matching $\mu_1 \veebar_{\mathcal{W}} \mu_2$ for each worker follows from its definition. We next show that the same conclusion holds for the firms. Let $f \in \mathcal{F} $ and $\mu_1 \veebar_{\mathcal{W}} \mu_2(f)=S_1\cup S_2\cup S_3$ such that
\smallskip

$\begin{tabular}{l} $S_1=\{w\in \mu_1 \veebar_{\mathcal{W}} \mu_2(f): w\in  \mu_1(f)$ {and} $\mu_1(w)\not=\mu_2(w)\}$ \medskip \\ 
 $S_2=\{w\in \mu_1 \veebar_{\mathcal{W}} \mu_2(f): w\in  \mu_2(f)$ {and} $\mu_1(w)\not=\mu_2(w)\}$\medskip \\  $S_3=\{w\in \mu_1 \veebar_{\mathcal{W}} \mu_2(f):  \mu_1(w)=\mu_2(w)=f\}$ \end{tabular}$.

\smallskip

 Since all workers  $w\in S_1$ prefer $f$ to $\mu_2(w)$,  $S_1\subseteq W_{f,\mu_2}$. As $\mu_2\in  FQS(F,\mathcal{W},\mathbf{P})$, we have that $\mu_2(f)\subseteq Ch(\mu_2(f)\cup S_1,P(f))$. Also  $S_2\cup S_3\subseteq \mu_2(f)$, it follows that $S_2\cup S_3\subseteq Ch(\mu_2(f)\cup S_1,P(f))$, implying   that, $S_2\cup S_3\subseteq Ch((S_2\cup S_3)\cup S_1,P(f))=Ch(\mu_1 \veebar_{\mathcal{W}} \mu_2(f),P(f))$, by the substitutability of $P(f)$. Now, as $S_2\subseteq W_{f,\mu_1}$,  the stability of  $\mu_1$ implies that $\mu_1(f)= Ch(\mu_1(f)\cup S_2,P(f))$. Also, as  $S_1\cup S_3\subseteq \mu_1(f)$, we have that $S_1\cup S_3\subseteq Ch(\mu_1(f)\cup S_2,P(f))$, then the   substitutability of $P(f)$ implies that $S_1\cup S_3\subseteq Ch((S_1\cup S_3)\cup S_2,P(f))=Ch(\mu_1 \veebar_{\mathcal{W}} \mu_2(f),P(f))$. We conclude that $\mu_1 \veebar_{\mathcal{W}} \mu_2(f)=S_1\cup S_2\cup S_3\subseteq Ch(\mu_1 \veebar_{\mathcal{W}} \mu_2(f),P(f))$, completing the proof that $\mu_1 \veebar_{\mathcal{W}} \mu_2$ is  individually rational for each  firm. 

\smallskip

We claim that the following equality 
\begin{equation}
|\mu_1 \veebar_{\mathcal{W}} \mu_2(f)|=|\mu_1(f)|  \label{7.3}
\end{equation}
holds for any firm $f \in  \mathcal{F}$. Assume that there exists al least one firm  $f\in  \mathcal{F}$ such that $|\mu_1(f)|<|\mu_1 \veebar_{\mathcal{W}} \mu_2(f)|$. Then, we can find $w\in \mu_1 \veebar_{\mathcal{W}} \mu_2(f)\setminus \mu_1(f)$. So,  $w\in \mu_2(f)$ and $fP(w) \mu_1(w)$. By  individual rationality of $\mu_2$ for the firms and   $q_f$-separability of $P(f)$, $wP(f)\emptyset$. As $|\mu_1(f)|<|\mu_1 \veebar_{\mathcal{W}} \mu_2(f)|\leq q_f$ (the last inequality is implied by $\mu_1 \veebar_{\mathcal{W}} \mu_2$ is  individually rational) $q_f$-separability of $P(f)$ implies that  $w \in Ch(\mu_1(w)\cup \{w\},P(f))$. So,  $(w,f)$ is a blocking pair of $\mu_1$, a contradiction. Therefore,
$|\mu_1 \veebar_{\mathcal{W}} \mu_2(f)|\leq|\mu_1(f)|$ for all $f \in  F$. Assume that there exists $f\in  \mathcal{F}$ with the property that $|\mu_1 \veebar_{\mathcal{W}} \mu_2(f)|<|\mu_1(f)|$. Then
\begin{center}
$\displaystyle\sum_{f \in  F}|\mu_1 \veebar_{\mathcal{W}} \mu_2(f)|<\displaystyle\sum_{f \in  F}|\mu_1(f)|$, 
\end{center}
which implies that there exists   $w \in \displaystyle\bigcup_{f \in {\cal F}} \ \mu_1(f)\setminus \displaystyle\bigcup_{f \in {\cal F}} \ \mu_1 \veebar_{\mathcal{W}} \mu_2(f)$. Hence, we have that there exists two firms $f'$ and $f^*$, such that $w\in \mu_1(f')$ and  $\mu_2(w)=f^*$ o $\mu_2(w)=w$. Then, by the definition of $\veebar_{\mathcal{W}}$, we have either $w\in \mu_1 \veebar_{\mathcal{W}} \mu_2(f')$ or $w\in \mu_1 \veebar_{\mathcal{W}} \mu_2(f^*)$, which contradicts the fact that $w\not\in \displaystyle\bigcup_{f \in {\cal F}}\mu_1 \veebar_{\mathcal{W}} \mu_2(f)$. Thus, $|\mu_1 \veebar_{\mathcal{W}} \mu_2(f)|=|\mu_1(f)|$ for all $f \in  F$. 

\smallskip

To finish with the proof that $\mu_1 \veebar_{\mathcal{W}} \mu_2 $ is a stable matching, assume that the pair $(w,f)$ blocks $\mu_1 \veebar_{\mathcal{W}} \mu_2 $,  that is,   
\begin{equation}\label{block-1}
 fP(w) \mu_1 \veebar_{\mathcal{W}} \mu_2(w) \ \ and \ \ w\in Ch(\mu_1 \veebar_{\mathcal{W}} \mu_2(f),P(f)) 
\end{equation}

We distinguish between the following two cases: 
\smallskip

 \noindent Case 1: $|\mu_1 \veebar_{\mathcal{W}} \mu_2(f)|<q_f$.
By condition  (2)    we have that  $fP(w) \mu_1 \veebar_{\mathcal{W}} \mu_2(w) R(w) \mu_1(w)$.  Furthermore, since    $|\mu_1(f)|=|\mu_1 \veebar_{\mathcal{W}} \mu_2(f)|<q_f$ (by condition  (1)),  $wR(f) \emptyset$ and  $q_f$-separability of $P(f)$  we have that $w\in Ch(\mu_1(f)\cup \{w\},P(f))$. Then, the pair $(w,f)$
 blocks   $\mu_1$, a contradiction.

\smallskip

\noindent Case 2: $|\mu_1 \veebar_{\mathcal{W}} \mu_2(f)|=q_f$.  
Then, there exists  $w_1\in \mu_1 \vee_{\mathcal{W}} \mu_2(f)$ such that
\begin{equation}\label{trab-desp}
w_1\not\in C(\mu_1 \veebar_{\mathcal{W}} \mu_2(f)\cup \{w \},P(f)).  
\end{equation}
First, we assume that $w_1\in \mu_1(f)$. We claim that the following equality
\begin{equation}\label{igualdad}
Ch(\mu_1\veebar_\mathcal{W} \mu_2(f) \cup \{w\} \cup \mu_1(f),P(f))=\mu_1(f)
\end{equation}
holds. Assume that there exists  $w' \in [Ch(\mu_1\veebar_\mathcal{W} \mu_2(f) \cup \{w\} \cup \mu_1(f),P(f))]\setminus [\mu_1(f)]$. Then either $w'=w$ in which case, by condition (%
\ref{block-1}) and the substitutability of $P(f)$  the
pair $(w,f)$ also blocks  $\mu_1$  or else $w'\not=w$ , implying that, $w'\in Ch(\mu_1\veebar_\mathcal{W} \mu_2(f)  \cup \mu_1(f),P(f))$ by the substitutability of $P(f)$. Therefore, and again by the
substitutability of $P(f)$ we have that $w'\in Ch(\{w'\}  \cup \mu_1(f),P(f))$. But since $w'\in \mu_1\veebar_\mathcal{W} \mu_2(f)\setminus \mu_1(f)$  $fP(w') \mu_1(w')$,  which
implies that the pair $ (w',f)$ blocks $\mu_1$. Therefore, condition (\ref{igualdad}) holds. Applying again the
assumption that $P(f)$ is substitutable, we have that 
$w_1\in Ch(\mu_1\veebar_\mathcal{W} \mu_2(f) \cup \{w\} \cup \{w_1\},P(f))$
which contradicts \ref{trab-desp} since $w_1\in \mu_1\veebar_\mathcal{W} \mu_2(f)$. Secondly, assume that $w_1\in \mu_2(f)$. Let
 $S=\{w'\in \mu_1 \veebar_{\mathcal{W}} \mu_2(f): w'\in \mu_1(f)$ y $\mu_1(w')\not=\mu_2(w') \}$, then $S\subseteq W_{f,\mu_2}$ . By the firm quasi-stable of $\mu_2$,    $\mu_2(f) \subseteq Ch(S\cup \{w\} \cup \mu_2(f),P(f))$. Further the equality 
 $\mu_1 \veebar_{\mathcal{W}} \mu_2(f) \cup \{w \} \cup \mu_2(f) = S\cup \{w \} \cup \mu_2(f)$ implies that 
$\mu_2(f) \subseteq Ch(\mu_1\veebar_\mathcal{W} \mu_2(f) \cup \{w\} \cup \mu_2(f),P(f))$. Therefore,  by the substitutability of $P(f)$ we have that $w_1 \in  Ch(\mu_1\veebar_\mathcal{W} \mu_2(f) \cup \{w\} \cup \{w_1\} ,P(f))$ which contradicts \ref{trab-desp} since $w_1 \in \mu_1\veebar_\mathcal{W} \mu_2(f)$.
\end{proof}

\medskip

The next corollary   will be very useful in our development.

\begin{corollary}\label{lema5}
Let $\mu'\in FQS (\mathcal{F}, \mathcal{W},\mathcal{P})$ and $w\in \mathcal{W}$. Then either  $\mu'(w)$ is achieva\-ble for $w$ or $\mu_\mathcal{F}(w)P(w)\mu'(w)  $
\end{corollary}
\begin{proof}
Assume  $\mu'(w)$  is not achievable for $w$ and let $\mu^* = \mu' \veebar_\mathcal{W} \mu_{\cal F}$. By Theorem \ref{s} $\mu^*\in S( \mathcal{F},\mathcal{W},\mathcal{P})$, then $\mu^*(w)$ is achievable for $w$. So, as $\mu^*(w)\in \{ \mu'(w),\mu_ F(w) \}$  and $\mu'(w)$ is not achievable for $w$,  $\mu^*(w)= \mu_\mathcal{F}(w)P(w) \mu'(w)$.
\end{proof}

Given an acceptable matching $\mu'$, define
\begin{center}
$S_\mathcal{W}(\mu')=\{\mu \in S(\mathcal{F},\mathcal{W},\mathcal{P}):\mu\succeq_\mathcal{W} \mu'\}$
\end{center}
that is, $S_\mathcal{W}(\mu')$ the set of stable matchings that the workers weakly prefer to $\mu'$.

In the following lemma  we prove that when $\mu'$ is a firm-
quasi-stable matching,  $S_\mathcal{W}(\mu')$ is a nonempty sub-lattice of  $S(\mathcal{F},\mathcal{W},\mathcal{P})$ under the partial order $\succeq_\mathcal{W}$, with lattice operators $\barwedge_\mathcal{W}$ and $\veebar_\mathcal{W}$.

\begin{lemma} \label{3}
Let $\mu'\in FQS(\mathcal{F},\mathcal{W},\mathcal{P})$. Then $S_\mathcal{W}(\mu')$ is a non-empty sub-lattice of $S(\mathcal{F},\mathcal{W},\mathcal{P})$.
\end{lemma}

\begin{proof}
By Theorem \ref{s}, $\mu' \veebar_{\mathcal{W}} \mu_\mathcal{W} \in S(\mathcal{F},\mathcal{W},\mathcal{P})$. and by the
optimality of $\mu_\mathcal{W}$  within $S(\mathcal{F},\mathcal{W},\mathcal{P})$  it follows that $\mu_\mathcal{W} \succeq_\mathcal{W} \mu' \veebar_{\mathcal{W}} \mu_\mathcal{W}$. Hence
$\mu_\mathcal{W} \succeq_\mathcal{W} \mu' \veebar_{\mathcal{W}} \mu_\mathcal{W}\succeq_\mathcal{W} \mu'$. Then  $\mu_\mathcal{W} \in S_\mathcal{W}(\mu')$, hence $S_\mathcal{W}(\mu')\not=\emptyset$.
Further, if $\mu_1,\mu_2 \in S_\mathcal{W}(\mu')$, then $\mu_1,\mu_2 \in S(\mathcal{F},\mathcal{W},\mathcal{P})$ and $\mu_1,\mu_2\succeq_\mathcal{W} \mu'$. By Theorem \ref{reti} $\mu_1 \veebar_\mathcal{W} \mu_2, \mu_1 \barwedge_\mathcal{W} \mu_2\in S(\mathcal{F}, \mathcal{W},\mathcal{P})$. Then by
the  definition of $\barwedge_\mathcal{W}$ and $\veebar_\mathcal{W}$ it follows that  $\mu_1 \veebar_\mathcal{W} \mu_2 \succeq_\mathcal{W} \mu' $ and $\mu_1 \barwedge_\mathcal{W} \mu_2 \succeq_\mathcal{W} \mu'$. So,  $\mu_1 \veebar_\mathcal{W} \mu_2 \in S(\mu')$ and $\mu_1 \barwedge_\mathcal{W} \mu_2 \in S(\mu')$, completing the proof that $S_\mathcal{W}(\mu')$ is a non-empty sub-lattice of $S(\mathcal{F},\mathcal{W},\mathcal{P})$.
\end{proof}

\medskip

\section[Restabilization after change in the population ]{Restabilization after change in the population}

Consider  a market which has achieved a stable matching, and
which is then disrupted by  change in the population (the retirement of some workers and/or the entry of new firms). The Set Offering (SO) Algorithm introduced by Cantala (2004) shows that   the market always reaches stability again  after experiencing disruptions which result in firm quasi-stable matchings.\footnote{The SO Algorithm is an   adapted version of the DA Algorithm where firms make offers.}.
We then  provide several characterizations of the stable matching obtained though  the SO Algorithm. 



Cantala (2004) establishes that by   using  a firm quasi-stable matching as input,  the sequence of tentative matching produced by the  SO  Algorithm is composed by firm quasi-stable matchings (note that this property will be token into account from now on wherever we describe the algorithm) and the output matching is stable.  
\medskip

 We consider a market $(\mathcal{F},\mathcal{W},\mathcal{P})$ and a firm  quasi-stable  matching. For each iteration, consider the set of workers who may want to join firm  $f$. This set, $A_f^{i-1}$, is the set of workers who have never rejected the offer of $f$ or have never been matched to it until
iteration  $i$, and who belong to an acceptable subset of workers. Each firm makes an offer to the workers in   $S_f^{i-1}$,   which belong to the  choice de $f$ among the set of available workers $A_f^{i-1}$ and their current workforce $\mu^{i-1}(f)$. Workers
accept the offer of their favorite firm between their present match and the firms which made them an offer. Given the new matching, we let firms make new offers if they wish, in the
next iteration. If no firm wants to make offers, the dynamic stops. Formally:

\medskip
\noindent{\bf The Set Offering  Algorithm}\\[0,2cm]
\noindent{\bf Input}\\[0,2cm]
A market $(\mathcal{F},\mathcal{W},\mathcal{P})$ and a matching $\mu$ \\[0,2cm]
{\bf Initialization}
\begin{itemize}
\item[(a)] $\mu^0=\mu$ y $i=1$.
\item[(b)]  For all $f \in  \mathcal{F}$ set $A_f^0=\{w\in \mathcal{W}:w\not\in \mu^0(f)$ and there exists $S\in 2^\mathcal{W}$ such that $|S|<q_f$ and $S\cup \{w\}P(f) S\}$.
\end{itemize}
{\bf Main Iteration} 
\begin{itemize}
\item[(1)] For all  $f\in \mathcal{F}$, define  $S^{i-1}_f=Ch(A_f^{i-1}\cup \mu^{i-1}(f),P(f))\setminus \mu^{i-1}(f)$.
\item[(2)] If there is no  $f\in  \mathcal{F}$ such that  $S_f^{i-1}\not=\emptyset$, stop with output $\mu^{i-1}$; otherwise  each firm $f$ makes offers to workers in $S_f^{i-1}$.
\item[(3)] For each  $w\in \mathcal{W}$ who received an offer at step $(2)$, let  $T_w^{i-1}=\{f\in  \mathcal{F}: f=\mu^{i-1}(w)$ o $w\in S^{i-1}_f\}$ and define $\mu^i(w)$ as the  $P(w)$-most preferred element in  $T_w^{i-1}$.
\item[]  For all  $w'\in \mathcal{W}$ who did not received any offer at step (2), $\mu^i(w')=\mu^{i-1}(w')$.
\item[(4)] Finally let  $A_f^i=A_f^{i-1}\setminus S^{i-1}_f$ or all firm $f$.
\item[(5)]$i=i+1$, go to $(1)$.
\end{itemize}

Let  $SO(\mu)$ be a the output matching of the SO algorithm  with input  mat\-ching $\mu$ and $[SO(\mu)](w)$ the match of worker $w$ at  $SO(\mu)$. 

\medskip


\begin{theorem}\label{of2}
Let $\mu$ be a matching in the market $(\mathcal{F},\mathcal{W},\mathcal{P})$. If $\mu$ is firm quasi-stable, then so are all intermediate matchings along the SO Algorithm; in particular, the output matching $SO(\mu)$ is stable.\footnote{Cantala (2004) uses the condition of $q$-substitutability on the preferences of the firms.}
\end{theorem}

In  order  to  characterize  the  output  of  the  SO  Algorithm   we  need  the following two lemmas, that records properties of the output of the SO Algorithm in terms of the input matching.


\begin{lemma}\label{1}
Let  $\mu'\in FQS(\mathcal{F},\mathcal{W},\mathcal{P})$. Then  $SO(\mu') \succeq_\mathcal{W} \mu'$.
\end{lemma}
\begin{proof}
Let $\mu'=\mu^0,\ \mu^1,...,\ \mu^k$,   be the sequence of distinct matchings gene\-rated by  the SO Algorithm with input $\mu'$ and let $w'$ be a  worker who received an offer at step 2 of the main iteration
in which  $\mu^i$  is generated. Then $ \mu^i(w)P(w)\mu^{i-1}(w)$ by definition of $T_w^{i-1}$ and for  all $w\in \mathcal{W}$ did  not  receive  any  offer  at  step  2  $ \mu^i(w)=\mu^{i-1}(w)$, hence $\mu^i\succeq_\mathcal{W} \mu^{i-1}$. By iterating this inequality we conclude that $\mu'=\mu^0 \preceq_\mathcal{W} \mu^1 \preceq_\mathcal{W}...\preceq_\mathcal{W} \mu^k=SO(\mu')$.
\end{proof}


\begin{lemma}\label{2}
Let $\mu' \in FQS(\mathcal{F},\mathcal{W},\mathcal{P})$ and  let $\mu^i$ be the i-th matchings generated by  the SO Algorithm with input $\mu'$ . If $\mu \in S(\mathcal{F},\mathcal{W},\mathcal{P})$ and $\mu \succeq_\mathcal{W} \mu'$, then $\mu \succeq_\mathcal{W} \mu^i$.
\end{lemma}
\begin{proof}
Let $\mu'=\mu^0,\ \mu^1,...,\ \mu^k$,   be the sequence of distinct matchings generated by  the SO Algorithm with input $\mu'$. We will show, by induction, that $\mu \succeq_\mathcal{W} \mu^i$ for each $i=0,...,k$. For $i=0$ the conclusion follows from the assumption  $\mu \succeq_\mathcal{W} \mu'$. Assume, by inductive hypothesis,   that $\mu(w)P(w) \mu^{i-1}(w)$ for all $w \in \mathcal{W}$ and consider  $\mu^i$. To get a contradiction, assume that there exits  $w'\in \mathcal{W}$ such that $\mu^i(w')P(w') \mu(w')$.
Consequently $\mu^i(w')=f'\in \mathcal{F}$, we get from $\mu(w')P(w')\mu^{i-1}(w')$ (our inductive assumption) and $\mu^{i-1}(w')P(w')\mu'(w')$ (established in the proof of Lemma \ref{1}) that  $f'=\mu^i(w')P(w')\mu'(w')$. Moreover, $w'\in Ch(\mu^i(f'),P(f'))$ (since  $w'\in \mu^i(f')$ and by Theorem \ref{of2} $\mu^i$ is firm quasi-stable matching), thus $w'\in A_{f'}^0$.
As $f'P(w')\mu^{i-1}(w')$, then $f'$ does not make an offer to $w'$ until iteration $i-1$. 
We conclude that
\begin{center}
 \hspace{4cm} $w'\in A_{f'}^{i-1}$    \hfill(3)
\end{center}
We claim that the following inclusion
\begin{center}
 \hspace{3cm} $ \mu(f_2)\subseteq   A_{f'}^{i-1}\cup \mu^{i-1}(f')$  \hfill (4) 
\end{center}
holds. Let  $w\in \mu(f')$ such that $w\not\in \mu^{i-1}(f')$. Since $f'P(w)\mu^{i-1}(w)$ (our inductive assumption) and $\mu^{i-1}(w)P(w)\mu'(w)$ (established in the proof of Lemma \ref{1}), we have that   $f'P(w)\mu'(w)$ and by stability of   $\mu$    $w\in Ch(\mu(f'),P(f'))$, then  $w\in A_{f'}^0$. Moreover,   since $f'P(w) \mu^{i-1}(w)$, then    $f'$ does not make an offer to  $w$ until iteration $i-1$. 
We conclude that $w\in A_{f'}^{i-1}$.   

\smallskip
 
By (3) and (4) $w'\cup \mu(f') \subseteq A_{f'}^{i-1}\cup \mu^{i-1}(f')$. Since $\mu^i(w')=f'$ then $w'\in Ch(A_{f'}^{i-1}\cup \mu^{i-1}(f'),P({f'}))$ and by the substitutability of $P(f')$   $w' \in Ch(w'\cup \mu(f'),P({f'}))$ . As $f'=\mu^i(w')P(w')\mu(w')$ we have that $(w',f')$ is a blocking pair for $\mu$, contradicting the assumption that $\mu$ is stable.
\end{proof}

\medskip

We next show that the output of every execution of the SO Algorithm with input a firm quasi-stable matching $\mu'$, is the worker-worst
stable matching in $S_\mathcal{W}(\mu')$.

\begin{theorem}\label{os}
Let $\mu'\in FQS(\mathcal{F},\mathcal{W},\mathcal{P})$. Then \begin{itemize}
    \item [1.] $SO(\mu')\in S_\mathcal{W}(\mu')$.
    \item [2.] If $\mu\in S_\mathcal{W}(\mu')$, then $\mu\succeq_\mathcal{W} SO(\mu')$
\end{itemize}
\end{theorem}

\begin{proof} The result is immediate from Theorem \ref{of2} and Lemmas \ref{1} and \ref{2}.
\end{proof}

\medskip

As $S_\mathcal{W}(\mu')$ is a nonempty sub-lattice of $S(\mathcal{F},\mathcal{W},\mathcal{P})$ thus, 
contains the matching $\barwedge_\mathcal{W}S_\mathcal{W}(\mu')$, which is the worker-worst stable matching in  $S_\mathcal{W}(\mu')$. We next show that this stable matching is the output  of the SO Algorithm with input $\mu'$.

\begin{theorem}\label{15}
Let  $\mu'\in FQS(\mathcal{F},\mathcal{W},\mathcal{P})$. Then  $SO(\mu')= \barwedge_\mathcal{W} S_\mathcal{W}(\mu')$.
\end{theorem}
\begin{proof}
The result is immediate from Theorem \ref{os}.
\end{proof}

\medskip

The next theorem gives a simple representation of the output of the $SO$ Algorithm  on firm-quasi-stable matchings in terms of the input matching and the firm-optimal stable matching  $\mu_{ F}$.

\begin{theorem}\label{supremo}
Let $\mu'\in FQS(\mathcal{F},\mathcal{W},\mathcal{P})$. Then $SO(\mu')=\mu'\veebar_\mathcal{W} \mu_{\mathcal{F}}$.
\end{theorem}

\begin{proof}
By Theorem \ref{s} and since $\mu'\veebar_\mathcal{W} \mu_{ \mathcal{F}}\succeq_\mathcal{W} \mu'$, we have that $\mu' \veebar_{\mathcal{W}} \mu_\mathcal{W} \in S_\mathcal{W}(\mu')$. Next, for $\mu\in S_\mathcal{W}(\mu') $ we have that $\mu \in S(\mathcal{F},\mathcal{W},\mathcal{P})$ y  $\mu\succeq_\mathcal{W} \mu'$, then  $\mu\succeq_\mathcal{W} \mu_{ \mathcal{F}}$, we conclude that $\mu\succeq_\mathcal{W} \mu'\veebar_\mathcal{W} \mu_{ \mathcal{F}}$. So, $\mu'\veebar_\mathcal{W} \mu_{\mathcal{F}}     $ is the worker-worst matching in $S_\mathcal{W}(\mu') $, that is $ \mu'\veebar_\mathcal{W} \mu_{\mathcal{F}}=\barwedge_\mathcal{W} S(\mu')$. By Theorem \ref{15}  $SO(\mu')=\mu'\vee_\mathcal{W} \mu_{{ \mathcal{F}}}$.
\end{proof}

\medskip

The next corollary provides another closed-form representation for the outcome of a worker under the output of the SO Algorithm.

\begin{corollary}\label{cor2}
Let $\mu'\in FQS(\mathcal{F},\mathcal{W},\mathcal{P})$ and $w \in  \mathcal{W}$. Then 
\begin{center}
$[SO(\mu')](w)=  \left\{\begin{tabular}{l} $\mu'(w)$ if $\mu'(w)$ is achievable for $w$ \\ $\mu_F(w)$ if $\mu'(w)$ is not achievable for $w$. \end{tabular}\right.$
\end{center}
\end{corollary}

\begin{proof}
Let $\mu=SO(\mu')$ and $w \in \mathcal{W}$. If $\mu'(w)$  is achievable
for  $w$,  then, because $\mu_{ \mathcal{F}}$ is the worker-worst stable matching, $\mu'(w)P(w) \mu_{ \mathcal{F}}(w)$.  Hence $\mu' \veebar_\mathcal{W} \mu_{ \mathcal{F}}(w)=\mu'(w)$ and by Theorem  \ref{supremo} $\mu(w)=\mu'(w)$. If, alternatively, $\mu'(w)$ is not achievable for $w$, then by Corollary \ref{lema5} $\mu_{ \mathcal{F}}(w) P(w) \mu'(w)$. Hence $\mu'\veebar_\mathcal{W} \mu_{ \mathcal{F}}(w)= \mu_{ \mathcal{F}}(w)$ and by Theorem \ref{supremo} $\mu(w)=\mu'\vee_\mathcal{W} \mu_{ \mathcal{F}}(w)= \mu_{ \mathcal{F}}(w)$. 
\end{proof}

\medskip

Corollary \ref{cor2} shows that when the input is a firm quasi-stable matching, the outcome of a  worker  under the SO Algorithm is determined by  their initial outcome and is independent of the initial outcome of the other agents. In particular, if the assignments of a worker under two firm quasi-stable matchings coincide, then the assignments of that worker coincide under the corresponding outputs of the SO Algorithm.

\section[The connection between the original stable matching and
the one obtained  in the new market]{The connection between the original stable matching and
the one obtained  in the new market}

In Section 3  we study the  firm quasi-stable matchings which can  arise from stable matchings following the creation of new jobs and/or the retirement of workers. The SO Algorithm restabilizes any firm quasi-stable matching, its restabilization process connects the stable matching in the original (pre-job creation and retirement) market with the stable matching achieved in the new market.  In the current section we use results of Section 4, regarding the characterization  of  the  outcome  of  the  Set  Offering algorithm, to explore the connection between the original stable matching and the one obtained after the restabilization process in the new market.

\medskip

The following definitions of Blum et al. (1997)  relate two arbitrary markets.
\medskip

Let $(\mathcal{F},\mathcal{W},\mathcal{P})$ and $(\mathcal{F'},\mathcal{W'},\mathcal{P'})$ be distinct markets. We say that the market $(\mathcal{F},\mathcal{W},\mathcal{P})$  is consistent with market $(\mathcal{F'},\mathcal{W'},\mathcal{P'})$ if the natural restrictions of $\mathcal{P}$ and $\mathcal{P'}$ to the set $( \mathcal{F}\cap \mathcal{F'})\cup (\mathcal{W}\cap \mathcal{W'})$ coincide, i.e., if for ${ \mathcal{F}}^*={\mathcal{F}}\cap { \mathcal{F'}}$ and $\mathcal{W}^*=\mathcal{W}\cap \mathcal{W'}$ the following conditions hold:

\begin{itemize}
\item[(1)] $A(\mathcal{F},\mathcal{W},\mathcal{P})\cap ( \mathcal{F}^* \times \mathcal{W}^*)= A(\mathcal{F'},\mathcal{W'},\mathcal{P'})\cap  (\mathcal{F}^* \times \mathcal{W}^*)$ , i.e., the set of accepta\-ble worker-firm coalitions is the same in both markets when one considers
only agents who belong to both markets;
\item[(2)] for each $f\in { F^*}$ and all $S,S'\subseteq \mathcal{W}^*$, $S P(f)f S'$ if and only if $S P'(f) S'$;
\item[(3)] for each $w\in \mathcal{W}^*$ and all $f,f'\in { F}^*$, $fP(w) f'$ if and only if $fP'(w)f'$.
\end{itemize}

We say that market $(\mathcal{F'},\mathcal{W'},\mathcal{P'})$ leads to market $(\mathcal{F},\mathcal{W},\mathcal{P})$, which is written as $(\mathcal{F'},\mathcal{W'},\mathcal{P'})\rightarrow (\mathcal{F},\mathcal{W},\mathcal{P})$ if
\begin{itemize}
\item[(1)]$(\mathcal{F},\mathcal{W},\mathcal{P})$ is consistent with $(\mathcal{F'},\mathcal{W'},\mathcal{P'})$ and 
\item[(2)]$\mathcal{W}\subseteq \mathcal{W}'$ y ${\mathcal{F}'}\subseteq {\mathcal{F}}$.
i.e., the set of positions of a first market is included in the
set of positions of a second market, and the converse for the population of workers.
\end{itemize}

We now relate two matchings in two distinct markets. Assume that $(\mathcal{F'},\mathcal{W'},\mathcal{P'})\rightarrow (\mathcal{F},\mathcal{W},\mathcal{P})$,  we say that $\mu'$ induces $\mu$ if for all $w\in \mathcal{W}$, $\mu(w)=\mu'(w)$,  which happens when new firms open and/or or workers retire. 

\medskip

Throughout this section we assume that $(\mathcal{F'},\mathcal{W'},\mathcal{P'})$ is a  market such that $(\mathcal{F'},\mathcal{W'},\mathcal{P'})\rightarrow (\mathcal{F},\mathcal{W},\mathcal{P})$, with the interpretation that $(\mathcal{F},\mathcal{W},\mathcal{P})$ is obtained from $(\mathcal{F'},\mathcal{W'},\mathcal{P'})$ through creation of new jobs and/or retirement of workers.

\medskip

Cantala (2004) showed that if $\mu'$ is a stable matching
for the original market $(\mathcal{F'},\mathcal{W'},\mathcal{P'})$, then $\mu$ induced by $\mu'$  is a firm quasi-stable matching in the  new market $(\mathcal{F},\mathcal{W},\mathcal{P})$.

\begin{lemma}\label{3.4}
Let two matchings be $\mu$ and $\mu'$ where $\mu'\in S(\mathcal{F'},\mathcal{W'},\mathcal{P'})$ and $\mu'$  induces $\mu$ in  $(\mathcal{F},\mathcal{W},\mathcal{P})$. Then $\mu\in FQS(\mathcal{F},\mathcal{W},\mathcal{P})$.
\end{lemma}

 Next we show that when the SO algorithm is used to restabilize the market, the entry of new firms and the retirement of workers cannot be good for any of the original firms, and cannot be bad for any of the workers who have not retired.

\begin{lemma}
Let two matchings be $\mu$ and $\mu'$ where $\mu'\in S(\mathcal{F'},\mathcal{W'},\mathcal{P'})$ and $\mu'$  induces $\mu$ in  $(\mathcal{F},\mathcal{W},\mathcal{P})$. Then $SO(\mu)\succeq_{\mathcal{W}}^{\mathcal{P}}\mu'$ and $\mu' \succeq_{\mathcal{F_B'}}^{\mathcal{P'}} SO(\mu)$ 
\footnote{Blair (1988), defines the partial ordering $\succeq_{\mathcal{F_B}}$
 on $(\mathcal{F},\mathcal{W},\mathcal{P})$ as
follows: $\mu _{1}\succeq_{\mathcal{F_B}}\mu _{2}\Leftrightarrow Ch(\mu_1(f) \cup \mu_2(f),P(f))=\mu_1(f)$ for all $f\in \mathcal{F}$. We sometimes add superscripts and write, for example $\succeq_{\mathcal{F_B}}^{\mathcal{P}}$ to emphasize dependence on particular preferences.}.
\end{lemma}

\begin{proof}
By Lemma \ref{3.4},  $\mu\in FQS(\mathcal{F'},\mathcal{W'},\mathcal{P'})$ and by Lemma \ref{1}, $SO(\mu) \succeq_\mathcal{W}^{\mathcal{P}} \mu$. As $\mu'(w)=\mu(w)$ for each $w\in \mathcal{W}\subseteq \mathcal{W}'$ have that $SO(\mu)\succeq_{\mathcal{W}}^{\mathcal{P}}\mu'$ as claimed. To prove the second inequality, assume by way of contradiction that $Ch([SO(\mu)](f) \cup \mu'(f),P'(f))\not =\mu'(f)$ for some $f \in \mathcal{F'}$.
Then there is  $w \in [SO(\mu)](f)\setminus \mu'(f)$ such that  $w \in Ch([(SO)(\mu)](f) \cup \mu'(f),P'(f))$ (since  $\mu'$ is individually rational) and   the substitutability of $P'(f)$ implies that  $w \in Ch(\mu'(f) \cup \{w\},P'(f))$.
By first conclusion    $f = [SO(\mu)](w)P(w) \mu'(w)$.  (there is no equality since  $\mu'(w)\not=f$. Further, since $(\mathcal{F},\mathcal{W},\mathcal{P})$ and $(\mathcal{F'},\mathcal{W'},\mathcal{P'})$  are consistent, the inequality $f P(w) \mu'(w)$ implies
that $f P'(w) \mu'(w)$. Hence, the pair $(f,w)$ is a blocking pair for $\mu'$ under $(\mathcal{F'},\mathcal{W'},\mathcal{P'})$, in contradiction to its asserted stability.
\end{proof}

\medskip

The following  lemma shows that each new firm ends up with its optimal achieva\-ble outcome.

\begin{lemma}\label{ultimafirma}
Let two matchings be $\mu$ and $\mu'$ where $\mu'\in S(\mathcal{F'},\mathcal{W'},\mathcal{P'})$ and $\mu'$  induces $\mu$ in  $(\mathcal{F},\mathcal{W},\mathcal{P})$ and $f\in \mathcal{F}\setminus \mathcal{F'} $. Then $[SO (\mu)](f)= \mu_ {\mathcal{F}}(f)$.
\end{lemma}
\begin{proof}
Let $f\in \mathcal{F}\setminus \mathcal{F'} $. Then $\mu(f)=\emptyset$ and by Lemma \ref{3.4},  $\mu\in FQS({ F},\mathcal{W},\mathbf{P})$. If $|\mu_{ \mathcal{F} }(f)|< q_f$  then by Theorems \ref{single} and \ref{of2},   
 $[SO(\mu)](f)=\mu_{\mathcal{F} }(f)$. If alternatively $|\mu_{\mathcal{F}  }(f)|= q_f$,  then by Theorems \ref{single} and \ref{of2},  $|[SO(\mu)](f)|=q_f$ and by Theorem \ref{supremo} \ $[SO (\mu)](f)\subseteq  \mu(f)\cup \mu_ { F}(f)=\mu_ { F}(f)$. Hence, $[SO (\mu)](f)= \mu_ {\mathcal{F}}(f)$.
\end{proof}

\medskip

The next theorem gives a closed-form representation of the outcome for each of the ¨original¨  workers under the output of the SO algorithm.

\begin{theorem}\label{repre-trabajadores}
Let two matchings be $\mu$ and $\mu'$ where $\mu'\in S(\mathcal{F'},\mathcal{W'},\mathcal{P'})$ 
and $\mu'$  induces $\mu$ in  $(\mathcal{F},\mathcal{W},\mathcal{P})$ and $w\in \mathcal{W}$. Then
$$[SO (\mu)](w)= max_w^{\mathcal{P}}\{\mu(w),\mu_{\mathcal{F}}(w)\} $$
\end{theorem}

\begin{proof}
By Lemma \ref{3.4},  $\mu\in FQS(\mathcal{F'},\mathcal{W'},\mathcal{P'})$, then Theorem \ref{supremo} implies that $SO(\mu)=\mu\veebar_\mathcal{W} \mu_{ F}$. Hence, by the definition  of   function $\mu\veebar_\mathcal{W} \mu_{ F}$,
$[SO (\mu)](w)= max_w^{\mathcal{P}}\{\mu(w),\mu_{\mathcal{F}}(w)\} $\end{proof}

\medskip

As a corollary of the above  representation we get the
monotonicity of the SO algorithm with respect  to the partial order $\succeq_{\mathcal{W}}$.

\begin{corollary}
Let $\mu'_1,\mu'_2 \in S(\mathcal{F'},\mathcal{W'},\mathcal{P'})$ and $\mu'_1,\mu'_2$  induces $\mu_1,\mu_2$ in  $(\mathcal{F},\mathcal{W},\mathcal{P})$, respectively. If $\mu'_1 \succeq_\mathcal{W}^{\mathbf{P}'} \mu'_2$ then $SO(\mu_1)\succeq_\mathcal{W}^{\mathcal{P}} SO(\mu_2)$
\end{corollary}

\begin{proof}
The conclusions of the corollary with respect to the workers are
imme\-diate from Theorem \ref{repre-trabajadores} and the fact that for each $w$  the inequality $\mu'_1(w)P'(w) \mu'_2(w)$ implies
that $\mu_1(w)P(w) \mu_2(w)$ (as  $(\mathcal{F},\mathcal{W},\mathcal{P})$ and $(\mathcal{F'},\mathcal{W'},\mathcal{P'})$  are consistent and $\mu_1(w)=\mu'_1(w)$, $\mu_2(w)=\mu'_2(w)$).



\end{proof}

\medskip

We conclude this section by showing what happens to the workers-optimal stable matching when there are changes in the population.

\begin{theorem}
Let $\mu_\mathcal{W'}$ and  $\mu_\mathcal{W}$ be  the
workers-optimal stable matchings, for $(\mathcal{F'},\mathcal{W'},\mathcal{P'})$ and   $(\mathcal{F},\mathcal{W},\mathcal{P})$, respectively. Then
\[
\mu_\mathcal{W}\succeq_\mathcal{W}^{\mathcal{P}} \mu_\mathcal{W'} \ \   \    , \  \ \mu_\mathcal{W'} \ \succeq_{\mathcal{F'_B}}^{\mathcal{P'}} \ \mu_\mathcal{W} .\ \   \  
\]
\end{theorem}

\begin{proof}
Assume  $\mu_\mathcal{W'}$ induces   $\mu'_\mathcal{W}$ in  $(\mathcal{F},\mathcal{W},\mathcal{P})$. By Lemma \ref{3.4},    $\mu'_\mathcal{W}\in FQS(\mathcal{F'},\mathcal{W'},\mathcal{P'})$, then theorem \ref{s} implies that $\mu'_\mathcal{W}\veebar_\mathcal{W} \mu_\mathcal{W} \in S(\mathcal{F},\mathcal{W},\mathcal{P})$, so  $\mu_\mathcal{W}\succeq_\mathcal{W}^{\mathcal{P}} \mu'_\mathcal{W}\veebar_\mathcal{W} \mu_\mathcal{W}  \succeq_\mathcal{W}^{\mathcal{P}} \mu'_\mathcal{W}$. Hence  $\mu_\mathcal{W}\succeq_\mathcal{W}^{\mathcal{P}} \mu_\mathcal{W'}$ (since  $\mu_\mathcal{W'}$ induces   $\mu'_\mathcal{W}$, $\mu'_\mathcal{W}(w)=\mu_\mathcal{W'}(w) $ for all $w\in \mathcal{W}$). Now assume by contradiction that $Ch(\mu_\mathcal{W}(f)\cup \mu_\mathcal{W'}(f),P'(f))\not=\mu_\mathcal{W'}(f)$ for some $f \in \mathcal{F'}$. Then there is $w\in \mu_\mathcal{W}(f)\setminus \mu_\mathcal{W'}(f)$ such that $w\in Ch(\mu_W(f)\cup \mu_\mathcal{W'}(f),P'(f))$ (since   $\mu_\mathcal{W'}$ is individually rational en $({F'},\mathcal{W}',\mathbf{P'})$) and   the substitutability of $P'(f)$ implies that $w\in Ch(\mu_\mathcal{W'}(f)\cup \{w\},P'(f))$. By first conclusion $f=\mu_\mathcal{W}(w)P(w)\mu_\mathcal{W'}(w)$ (there is no equality since $w\not\in \mu_\mathcal{W'}(f)$).  Further, since $(\mathcal{F},\mathcal{W},\mathcal{P})$ and $(\mathcal{F'},\mathcal{W'},\mathcal{P'})$  are consistent  the inequality $fP(w)\mu_\mathcal{W'}(w)$ implies
that $fP'(w)\mu_\mathcal{W'}(w)$. Hence, the pair $(f,w)$ is a blocking pair for $\mu_\mathcal{W'}$ en  $(\mathcal{F'},\mathcal{W'},\mathcal{P'})$.
\end{proof}

\bigskip

\noindent {\Large\textbf{References} }

\bigskip

\noindent Blair C. (1988).  The Lattice Structure of The set of Stable Matchings with Multiple Partners. Math. Operations Res. 13, 619-628.
\medskip

\noindent Blum Y., Roth A.E.,  Rothblum U.G., (1997).  Vacancy chains and  equilibration in senior-level labor markets. J. Econo. Theory 76, 362-411.
\medskip

\noindent  Cantala D. (2004). Restabilizing matching markets at senior level. Games Econ. Behav. 48, 1-17.
\medskip

\noindent  Cantala D. (2011). Agreement toward stability in matching markets. Rev. Econ. Design 15 (4), 293–316.
\medskip

\noindent  Gale D., Shapley L. (1962). Collage admissions and stability of  marriage. Amer. Math. Monthly 69, 9-15.
\medskip

\noindent  Kelso A.S., Crawford V.P., (1982). Job matching, coalition formation, and gross substitutes. Econometrica 
50, 1483-1504.
\medskip

 \noindent Ma\'inez R., Mass\'o J., Neme A., Oviedo J., (2000). Single agents and the set of many-to-one stable matchings. J. Econ.  Theory 91, 91-105.
\medskip

\noindent Mart\'inez R.,  Mass\'o J., Neme A., Oviedo J. (2001),  On The Lattice Structure of The Set of Stable Matchings for a Many-to-One Model. Optimization 50, 439-457.
\medskip

\noindent Roth, A. (1984). Stability and polarization of interests in job matching. Econometrica 52, 47–57.
\medskip

\noindent  Sotomayor M. (1996). A non-constructive elementary proof of the existence of stable marriages. Games Econ. Behav. 13, 135–137.
\medskip

\noindent  Wu Q., Roth (2018). The lattice of envy-free matchings. Games Econ. Behav. 109, 201–211.

\end{document}